\documentclass[conference,a4paper]{IEEEtran}
\usepackage[utf8]{inputenc} 
\usepackage[T1]{fontenc}
\usepackage{url}              
\usepackage{ifthen}          
\usepackage{cite}             
\usepackage[cmex10]{amsmath}  
\usepackage{mleftright}       
\mleftright                   
\usepackage{graphicx}         
\usepackage{booktabs}         
\usepackage{algorithm}
\usepackage{algorithmic}
\usepackage{epstopdf}
\usepackage{amsthm}
\usepackage{float}
\usepackage{subfigure}
\usepackage{cases}
\newtheorem{Le}{Lemma}
\newtheorem{theo}{Theorem}
\theoremstyle{definition}
\newtheorem{re}{Remark} 
\newcommand{\xir}{x_{i_{r-1}}\cdots x_{i_0}}
\newcommand{\xx}{x_{i_{\overline{r}-1}}\cdots x_{i_0}}
\newcommand{\xs}{x_{i_{r-q-1}}\cdots x_{i_0}}

\newcommand{\xk}{x_{i_{r-k-1}}\cdots x_{i_0}}
\newcommand{\zuhe}{ \begin{pmatrix} i_s \\ t \end{pmatrix}}
\newcommand{\MN}{\mathcal{N}}
\newcommand{\CA}{\mathcal{C}(\mathcal{I})}
\newcommand{\ol}{\overline{\ell}}

\begin{document}

\title{On the Weight Spectrum Improvement of Pre-transformed Reed-Muller Codes and Polar Codes} 

 \author{%
   \IEEEauthorblockN{Yuan Li\IEEEauthorrefmark{1}\IEEEauthorrefmark{2},
                     Zicheng Ye\IEEEauthorrefmark{1}\IEEEauthorrefmark{2},
                     Huazi Zhang\IEEEauthorrefmark{3},
                     Jun Wang\IEEEauthorrefmark{3},
                     Guiying Yan\IEEEauthorrefmark{1}\IEEEauthorrefmark{2},
                     and Zhiming Ma\IEEEauthorrefmark{1}\IEEEauthorrefmark{2}}
  \IEEEauthorblockA{\IEEEauthorrefmark{1}%
                     University of Chinese Academy of Sciences}
   \IEEEauthorblockA{\IEEEauthorrefmark{2}%
                     Academy of Mathematics and Systems Science, CAS }
   \IEEEauthorblockA{\IEEEauthorrefmark{3}%
                     Huawei Technologies Co. Ltd.}
    Email:  \{liyuan2018, yezicheng\}@amss.ac.cn, \{zhanghuazi, justin.wangjun\}@huawei.com,\\
           yangy@amss.ac.cn, mazm@amt.ac.cn 

 }

\maketitle

\begin{abstract}
Pre-transformation with an upper-triangular matrix (including cyclic redundancy check (CRC), parity-check (PC) and polarization-adjusted convolutional (PAC) codes)  improves the weight spectrum of Reed-Muller (RM) codes and polar codes significantly. However, a theoretical analysis to quantify the improvement is missing. In this paper, we provide asymptotic analysis on the number of low-weight codewords of the original and pre-transformed RM codes respectively, and prove that pre-transformation significantly reduces low-weight codewords, even in the order sense. For polar codes, we prove that the average number of minimum-weight codewords does not increase after pre-transformation. Both results confirm the advantages of pre-transformation.
\end{abstract}

\section{Introduction}
\label{Introduction}
Polar codes \cite{b1}, invented by Ar{\i}kan, are a great breakthrough in coding theory. As code length approaches infinity, the synthesized channels become either noiseless or pure-noise. Channel polarization occurs under successive cancellation (SC) decoding, which has a low complexity. However, the performance of polar codes under SC decoding is poor at short to moderate block lengths.

To boost finited-length performance, a successive cancellation list (SCL) decoding algorithm was proposed \cite{b2}. 
As list size $L$ increases, the performance of SCL decoding approaches that of maximum-likehood (ML) decoding. But the ML performance of polar codes is still inferior due to low minimum distance.
Consequently, concatenation of polar codes with CRC \cite{b3} and PC \cite{b4} were proposed to improve weight spectrum.
In Ar{\i}kan’s PAC codes\cite{b5}, convolutional precoding and RM rate-profiling were applied to approach binary input additive white Gaussian noise (BIAWGN) dispersion bound \cite{b7} under large list decoding\cite{b8}.

CRC-Aided (CA) polar, PC-polar, and PAC codes can be viewed as pre-transformed polar codes with upper-triangular transformation matrices\cite{b9}. In polar codes, frozen bits are all zeros, while in pre-transformed polar codes, traditional frozen bits are replaced by dynamically frozen bits \cite{s1}, whose value depends on previous bits.  It is proved that any upper-triangular pre-transformation does not reduce minimum distance\cite{b9}. In \cite{s2}, efficient recursive formulas were proposed to calculate the average weight spectrum of pre-transformed polar codes with polynomial complexity rather than exponential complexity. 

In this paper, we simplify the recursive formulas in \cite{s2} through the monomial representation of row vectors. From \cite{b9}\cite{s2}, low-weight codewords are induced by low-weight row vectors.
We further prove that, low-weight codewords are mainly induced by a small fraction of low-weight row vectors. Based on this discovery, we provide asymptotic analysis on the number of low-weight codewords of pre-transformed codes, and  quantitatively analyze the improvement of weight spectrum.   

This paper is organized as follows. In section II, we review polar codes and pre-transformed polar codes. In section III, we analyze the number of low-weight codewords of the original and pre-transformed RM codes respectively. Asymptotic analysis shows that low-weight codewords reduce significantly after pre-transformation. For polar codes, we prove that the average number of minimum-weight codewords does not increase after pre-transformation, as long as the code is decreasing \cite{b21}. Finally we draw some conclusions in section IV. 

\section{Background}
\label{Background}
\subsection{Polar Codes as Monomial Codes}
Let
$
F=\left(\begin{array}{ll}
1 & 0 \\
1 & 1
\end{array}\right)
$, $N=2^m$, and $F_N=F^{\otimes m}$.  Starting from $N$ binary-input discrete memoryless channels (B-DMC) $W$, we obtain $N$ synthetic channels $W_N^{(i)}$ after polarization. Polar codes can be constructed by selecting the indices of $K$ most reliable information sub-channels, i.e., $K$ row vectors of $F_N$, as information set $\mathcal{I}$. Density evolution (DE) algorithm\cite{b10}, Gaussian approximation (GA) algorithm\cite{b11} and the channel-independent polarization weight (PW) construction method\cite{b12} are efficient methods to find reliable sub-channels.

After determining the information set $\mathcal{I}$, its complement set is called the frozen set $\mathcal{F}$. Let 
$u_1^N=(u_1,u_2, \dots , u_N)$ be the bit sequence to be encoded. $K$ bits are inserted into $u_{\mathcal{I}}$,
and all zeros are filled into $u_{\mathcal{F}}$. Then the codeword $c_1^N$ is obtained by $c_1^N=u_1^NF_N$.

Polar codes can also be expressed as monomial codes\cite{b21} and the monomial set is denoted by
$$\mathcal{M}_m \overset{def}{=} \{ x_{m-1}^{a_{m-1}}\cdots x_{0}^{a_{0}} |({a_{m-1}},\dots,{a_{0}}) \in \mathbf F_2^m \}.$$ 
From this point of view, each row vector of $F_N$ corresponds to 
a monomial represented by $m$ binary variables \{$x_i$\}, $0 \leq i \leq m-1$. For instance, a monomial $f= x_{i_{r-1}}\cdots x_{i_{0}}$, where $0 \leq i_{0} <...<i_{r-1} \leq m-1$, the degree of $f$ is $r$, denoted by $deg(f)$. Denote $I_f = i$, if the monomial $f$ represents the $i$-$th$ row vector of $F_N$.

To be specific, let $(a_{m-1},...,a_0)$ be the binary representation of $N-i$, i.e., $N-i = \sum\limits_{j=0}^{m-1}2^ja_j$, then the $i$-$th$ row vector of $F_N$ can be represented by  monomial $ x_{m-1}^{a_{m-1}}\cdots x_0^{a_0}$. For example, the monomial representation of $F_8$ is shown in Fig. \ref{fig1}.

\begin{figure}[htbp]
\centerline{\includegraphics[width = .35\textwidth]{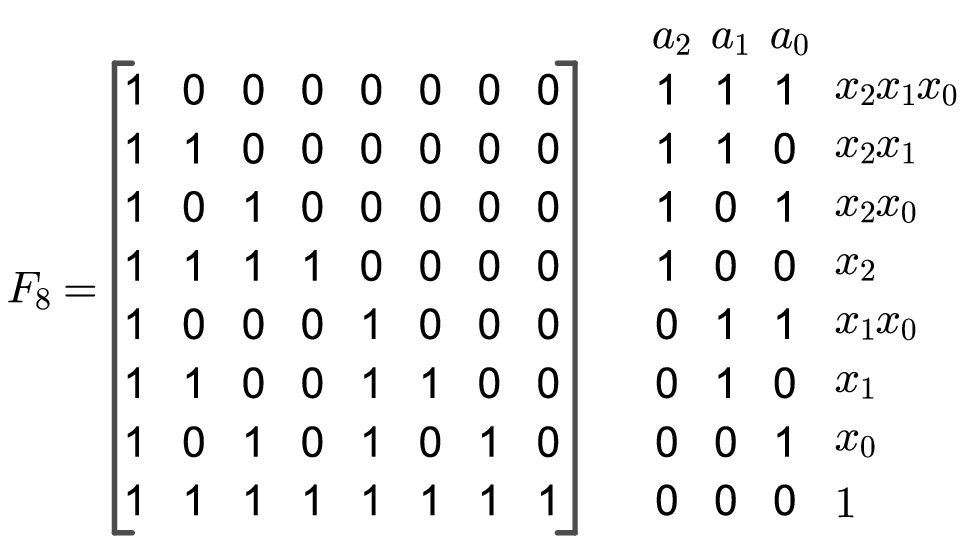}}
\caption{The monomial representation of $F_8$.}
\label{fig1}
\end{figure}

From now on, the $i$-$th$ row vector of $F_N$ and the corresponding monomial $x_{m-1}^{a_{m-1}}\cdots x_0^{a_0}$ are used interchangeably because they refer to the same thing.

\subsection{Decreasing Monomial Codes}
It was revealed in \cite{b21} and \cite{b31} that
the reliability of synthetic channels follows a partial order ``$\preceq$''. If $f,g\in \mathcal{M}_m$, $g \preceq f$ means $g$ is universally more reliable than $f$. 
For monomials of the same degree, partial order is defined as 
$$x_{i_{r-1}}\cdots x_{i_0} \preceq x_{j_{r-1}} \cdots x_{j_0} \iff i_s \leq j_s, \ 0 \leq s \leq r-1,$$ 
and for monomials of different degrees
$$g\preceq f \iff \exists f^*\mid f, \  deg(f^*)=deg(g), \  \text{and} \ g\preceq f^*.$$

Denote the monomial code with information set $\mathcal{I}$ by $\CA$, $\CA$ is a decreasing monomial code if $\mathcal{I}$ satisfies partial order, i.e.,
$$
\forall f \in \mathcal{I} \ \text{and} \ g\in \mathcal{M}_m,\ \text{if} \ g \preceq f \ \text{then} \  g \in \mathcal{I}. 
$$

For example, the information set of RM$(m,r)$ consists of monomials with degree no larger than $r$. By definition, RM codes are decreasing monomial codes. 

\subsection{Pre-Transformed Polar Codes}

\begin{equation*}
T
=\begin{bmatrix}
1  &  T_{12}  & \cdots\ &T_{1N}\\
0  &  1  & \cdots\ & T_{2N}\\
 \vdots   & \vdots & \ddots  & \vdots  \\
 0 & 0  & \cdots\ & 1\\
\end{bmatrix}.
\end{equation*}

Let $G_N = TF_N$ be the generator matrix of pre-transformed polar codes, where $T$ is an upper triangular pre-transformation matrix defined above. The codeword of the pre-transformed polar codes is given by $c_1^N=u_1^NG_N=u_1^NTF_N$. In PAC codes, the pre-transformation matrix $T$ is a Toeplitz matrix. 

In \cite{r4}, a deterministic algorithm for computing the weight spectrum of any given pre-transformed polar code was proposed. However, the computational complexity is still exponential.  In \cite{s2}, an efficient algorithm was proposed to compute the average weight spectrum of pre-transformed polar codes with polynomial complexity. The code ensemble assumes that $T_{ij}$, $1\leq i < j \leq N$ are $i.i.d.$  $Bernoulli(\frac{1}{2})$ $r.v.$.

\section{The number of low-weight codewords}
In this section, we provide asymptotic analysis on the number of low-weight codewords of the original and pre-transformed RM codes respectively. We prove that, for decreasing polar codes, the average number of minimum-weight codewords after pre-transformation  is no larger than that of the original codes.
\subsection{Notations and definitions}
In this paper, $\log x$ is the base-2 logarithm of $x$, $\lceil x \rceil$ is the ceiling function of $x$, and $w(x)$ is the Hamming weight of $x$. The entropy function $h(x) = -x\log x -(1-x)\log(1-x)$, $0 < x < 1$, and $|\mathcal{S}|$ is the cardinality of set $\mathcal{S}$. To characterize asymptotic results, we define the following notations\footnote{ $f(x) \leq O(g(x))$ $(f(x) \geq \Omega(g(x)))$, if $\limsup\limits_{x \rightarrow \infty}\frac{f(x)}{g(x)} < +\infty$ $\left(\liminf\limits_{x \rightarrow \infty}\frac{f(x)}{g(x)} > 0\right)$, where $g(x) > 0$. $f(x) = \Theta(g(x))$, if $0 <\liminf\limits_{x \rightarrow \infty}|\frac{f(x)}{g(x)}| \leq \limsup\limits_{x \rightarrow \infty}|\frac{f(x)}{g(x)}| < + \infty$, $f(x) = o(g(x))$, if $\lim\limits_{x \rightarrow \infty}|\frac{f(x)}{g(x)}| = 0$.}.

Let $f_N^{(i)}$ be the $i$-$th$ row vector of $F_N$, and $g_N^{(i)}$ be the $i$-$th$ row vector of $G_N$. Information set $\mathcal{I} = \left\{ I_1,I_2,\dots,I_K \right\}$, where $I_1 < ... < I_K$. The number of codewords with Hamming weight $d$ of polar/RM codes is denoted by $N(d)$, and the number of codewords with Hamming weight no larger than $d$ is denoted by $A(d)$. The minimum distance is denoted by $d_{min}$. The corresponding number of codewords of  pre-transformed codes with transformation matrix $T$ is denoted by $N(d,T)$ and $A(d,T)$, respectively. The average number is denoted by $E(N(d,T))$ and $E(A(d,T))$, where the expectation is with respect to random pre-transformation matrix $T$, and we assume $T_{ij}$, $1\leq i < j \leq N$ are $i.i.d.$  $Bernoulli(\frac{1}{2})$ $r.v.$.

Let $P(m,i,d) \overset{def}{=} P\left(w\left(g_{2^m}^{(i)}\right)=d\right)$ be the probability that the $i$-$th$ row vector of $G_N$ has Hamming weight $d$. By \cite[Lemma 2]{s2}, $P(m,i,d)=0$, if $d < w(f_N^{(i)})$, i.e., pre-transformation does not reduce the Hamming weight of row vectors.  According to \cite[Lemma 1]{s2}, the probability that the codeword $c_1^N = u_1^N G_N$ has Hamming weight $d$ is equal to $P(m,I_j,d)$, as long as $u_{I_j}$ is the first non-zero bit in $u_1^N$, thus we can combine the weight-$d$ codewords induced by these $2^{K-j}$ codewords in $E\left( N(d,T)\right)$ whose first non-zero bit is $u_{I_j}$. 
Let $N\left(m,I_j,d\right)\overset{def}{=}2^{K-j}P\left(m,I_j,d\right)$, according to \cite[eq.(7)]{s2},
\begin{align}
\label{eq1}
E\left( N(d,T)\right) & = \sum_{\mbox{\tiny$\begin{array}{c}
1 \leq j \leq K\\
w(f_{I_j}) \leq d\end{array}$}}2^{K-j}P \left(m,I_j,d \right) \notag \\
& = \sum_{\mbox{\tiny$\begin{array}{c}
1 \leq j \leq K\\
w(f_{I_j}) \leq d\end{array}$}}N \left(m,I_j,d \right), 
\end{align}
where $K-j$ is the number of information bits whose indices are greater than $I_j$. As explained above, $N\left(m,I_j,d\right)$ is the number of weight-$d$ codewords where $u_{I_j}$ is the first non-zero bit in the encoded bit sequence $u_1^N$. We call $N\left(m,I_j,d\right)$ the number of weight-$d$ codewords induced by the $I_j$-$th$ row vector. From (\ref{eq1}), all weight-$d$ codewords are induced by row vectors $f_N^{(i)}$ with weight no larger than $d$.  Therefore, when analyzing the number of weight-$d$ codewords in pre-transformed codes, we only need to consider the row vectors with weight no larger than $d$.

For convenience, we use $P(m,\xir,d)$ and $N(m,\xir,d)$ instead of $P(m,i,d)$ and $N(m,i,d)$ when $\xir$ represents the $i$-$th$ row vector of $F_N$.

\subsection{Low-weight codewords of RM codes}
In this section, we analyze low-weight codewords with Hamming weight within a constant multiple of minimum distance. We provide asymptotic analysis on the number of codewords in RM$(m,r)$ with Hamming weight no larger than $2^{m-r+k}$, where $k$ is a non-negative integer. The proof idea of Theorem \ref{thy} follows from \cite{y1} and \cite{y2}.

\begin{theo} 
\label{thy}
Assume $0< \alpha_1 < \frac{r}{m} < \alpha_2 <1$, where $\alpha_1$, $\alpha_2$ are constants, 
\begin{align}
\label{eqy1}
\Omega(m^{k+1}) \leq \log A\left(2^{m-r+k}\right) \leq O(m^{k+2}).
\end{align}
\end{theo}

\begin{proof}
The proof is in Appendix \ref{ap1}.
\end{proof}

\begin{re}
\label{rey}
When $k=0$, $ \log A\left(2^{m-r}\right) = \Theta(m^2)$ \cite{b21} reaches the upper bound of (\ref{eqy1}), and when $k=1$, $ \log A\left(2^{m-r+1}\right) = \Theta(m^2)$ \cite{y3} reaches the lower bound of (\ref{eqy1}) .
\end{re}

\subsection{Minimum-weight codewords of pre-transformed RM codes}
\label{se1}
According to (\ref{eq1}), the average number of minimum-weight codewords of pre-transformed RM$(m,r)$ is  
\begin{align}
\label{eq2}
&E\left( N(2^{m-r},T)\right)  \notag \\  
 =&\sum_{0 \leq i_0 <...<i_{r-1} \leq m-1} N(m,\xir,2^{m-r}).
\end{align}
Thus we first analyze the number of minimum-weight codewords induced by $\xir$.

\begin{Le}
\label{le1}
In RM($m,r$), the number of information bits whose indices are greater than $I_{\xir}$ is
$\sum\limits_{s=0}^{r-1}\sum\limits_{t=0}^{s+1}\zuhe$, and 
\begin{align}
\label{eq3}
\log P(m,\xir,2^{m-r}) = \sum_{s=0}^{r-1}\left(2^{i_s-s}-2^{i_s}\right).
\end{align} 
Thus $\log N(m,\xir,2^{m-r}) =$
\begin{align}
\label{eq4}
 \sum_{s=0}^{r-1}\left(2^{i_s-s}-2^{i_s}+\sum\limits_{t=0}^{s+1}\zuhe\right).
\end{align}
\end{Le}

\begin{proof}
We prove Lemma \ref{le1} via induction on $m$, the proof is in Appendix \ref{ap2}. 
\end{proof}

\begin{re}
In Lemma \ref{le1}, we simplify the recursive formulas in \cite[Theorem 1]{s2} through the monomial representation of the $i$-$th$ row vector of $F_N$, this simplified form is convenient for the further theoretical analysis. As seen, $P(m,\xir,2^{m-r})$ holds for all sub-channel selections, thus (\ref{eq3}) will also apply to polar codes.
\end{re}
Based on Lemma \ref{le1}, we provide asymptotic analysis on $N(m,\xir,2^{m-r})$ as well as $E\left( N(2^{m-r},T)\right)$.

\begin{theo}
\label{th1}
\begin{align}
\label{eq7}
\log N(m,\xir,2^{m-r}) \leq  \notag \\
\begin{cases}
0 & i_{r-1} \geq r+3, and \\ 
& r \ \text{sufficiently large}; \\
2r+3 & i_{r-1} = r+2; \\
3r & i_{r-1} \leq r+1.
\end{cases}
\end{align}
Assume $ m-r \geq 2$, $\frac{r}{m}>\gamma$, where $\gamma > 0$ is a constant, then
\begin{align}
\label{eq11}
3r \leq \log E\left( N(2^{m-r},T) \right) \leq 3r+O(\log r).
\end{align}
\end{theo}
\begin{proof}
The proof is in Appendix \ref{ap3}. We briefly introduce the proof outline below.

Firstly, we prove (\ref{eq7}) by \emph{Step 1-2}.

\emph{Step 1}, to further calculate (\ref{eq4}), let
\begin{align}
\label{eq8}
\MN(i_s,s) & = 2^{i_s-s}-2^{i_s}+\sum\limits_{t=0}^{s+1}\zuhe \notag \\
&= 2^{i_s-s} - \sum\limits_{t=0}^{i_s-s-2}\zuhe,
\end{align}
we prove $\MN(i_s,s) \leq 0$ if $s \geq 1$, $i_s-s \geq 3$. We analyze $\MN(i_s,s)$ according to the value of $s$, the proof is mainly based on the estimation of combinatorial number.

\emph{Step 2}, since $\log N(m,\xir,2^{m-r})=\sum\limits_{s=0}^{r-1} \MN(i_s,s)$, based on \emph{Step 1}, we analyze $\log N(m,\xir,2^{m-r})$ with respect to $i_{r-1}$, the proof details can be found in Appendix \ref{ap3}.

Next, we prove (\ref{eq11}) by \emph{Step 3}.

\emph{Step 3}, we divide the sum terms in (\ref{eq2}) into three parts according to $i_{r-1}$: $\sum\limits_{i_{r-1} \geq r+3}N(m,\xir,2^{m-r})$, \\ $\sum\limits_{i_{r-1} = r+2} N(m,\xir,2^{m-r})$ and \\ $\sum\limits_{i_{r-1} \leq r+1} N(m,\xir,2^{m-r})$. From \emph{Step 2}, the first term converges to zero, and the second term is negligible compared to the third term. Thus the minimum-weight codewords are mainly induced by $N(m,\xir,2^{m-r})$ with $i_{r-1} \leq r+1$, so 
\begin{align}
&E(N(2^{m-r},T)) \approx  \sum_{i_{r-1} \leq r+1} N(m,\xir,2^{m-r}) \notag \\
 \leq &|\{\xir, i_{r-1} \leq r+1\}|2^{3r} = \begin{pmatrix} r+2 \\ 2 \end{pmatrix}2^{3r}. 
\end{align}
\end{proof}

\begin{re}
\label{rer1}
Since we can not efficiently calculate the weight spectrum of specific pre-transformed codes, we analyze the average weight spectrum of the code ensemble defined by the random pre-transformation matrix. 

The results on the average weight spectrum are significant in two aspects. On the one hand, there exist good codes with minimum-weight codewords no larger than the average. 
On the other hand, numerical results confirm that, the actual number of minimum-weight codewords is usually very close to the average, i.e., has small variance. In practice, this means that most random pre-transformation matrices are good.
\end{re}

\begin{re}
\label{re2}
The $\begin{pmatrix} r+2 \\ 2 \end{pmatrix}$  monomials $\xir$ with $i_{r-1} \leq r+1$ induce the majority of miminum-weight codewords of pre-transformed RM codes, which is a tiny part of  $\begin{pmatrix} m \\ r \end{pmatrix}$ monomials with degree $r$. It implies that in pre-transformed polar codes, the minimum-weight codewords are mainly induced by a small fraction of monomials. For example, in RM$(9, 2)$, the 498-$th$ row vector $x_4x_3x_2$ satisfies $i_{r-1} \leq r+1$, its corresponding binary representation is $(0,0,0,0,1,1,1,0,0)$. Monomials $\xir$ with $i_{r-1} \leq r+1$ share similar characteristics: they are at the bottom of $F_N$ and have high reliability among monomials with degree $r$.
\end{re}

In Fig. \ref{fig2}, we display the number of minimum-weight codewords in RM codes and pre-transformed RM codes on the logarithm domain. The example has code rate $R=0.5$, and the average number of minimum-weight codewords is approximately $2^{3r}$ in the order sense. In contrast, the number before pre-transformation is $2^{\Theta (m^2)}$. In other words, the logarithm scaling of minimum-weight codewords drops from quadratic growth to linear growth after pre-transformation. The result proves that pre-transformation can reduce minimum-weight codewords significantly, even in the order sense. This also partly explains the gain of PAC codes (a special case of pre-transformed RM codes) over RM codes.

\begin{figure}[htbp]
\centerline{\includegraphics[width = .32\textwidth]{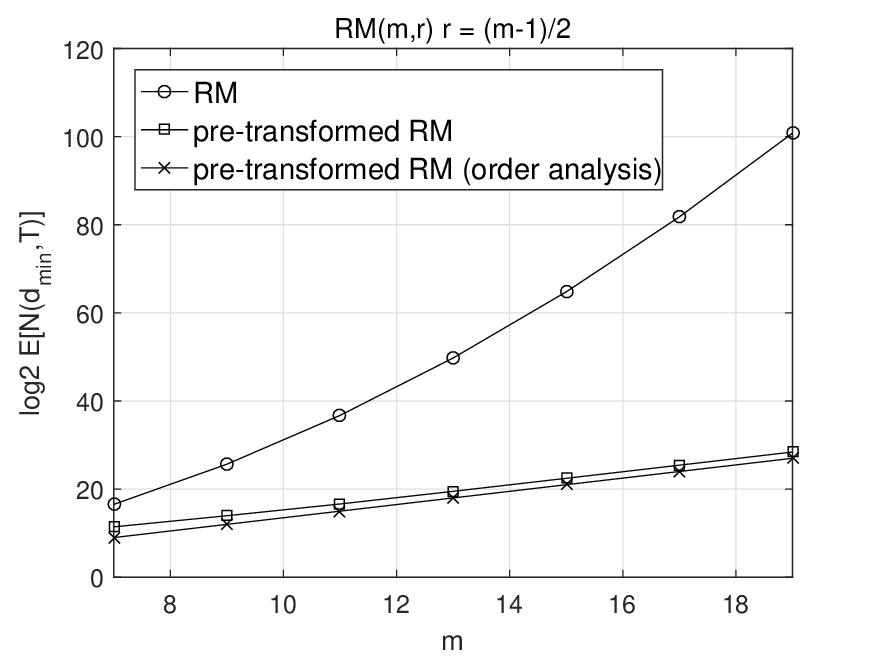}}
\caption{Logarithm scaling of the number of minimum-weight codewords in RM codes and pre-transformed RM codes.}
\label{fig2}
\end{figure}

\subsection{Low-weight codewords of pre-transformed RM codes}
In this section, we analyze low-weight codewords with Hamming weight within a constant multiple of minimum distance. We provide asymptotic analysis on the number of codewords in per-transformed RM$(m,r)$ with Hamming weight no larger than $2^{m-r+k}$, where $k$ is a positive integer. 

According to (\ref{eq1}), we only need to consider row vectors with weight no larger than $2^{m-r+k}$, or equivalently, monomials with degree at least $r-k$. For monomials with degree $r-q$, where $0 \leq q < k$, their corresponding row vectors have weight $2^{m-r+q} < 2^{m-r+k}$, thus they induce codewords with weight from $2^{m-r+q}$ to $2^{m-r+k}$.  Let 
\begin{align}
\label{eql0}
& A(m,\xs,2^{m-r-k}) \notag \\
=&\sum_{d'=2^{m-r+q}}^{2^{m-r+k}} N(m,\xs,d'),
\end{align}
$A(m,\xs,2^{m-r-k})$ is the number of codewords induced by $\xs$ with weight no larger than $2^{m-r+k}$. For monomials with degree $r-k$, their corresponding row vectors have weight exactly $2^{m-r+k}$, thus we only need to consider the number of weight-$2^{m-r+k}$ codewords induced by $\xk$. Therefore, we have
\begin{align}
\label{eql1}
&E(A(2^{m-r+k},T)) \notag \\
 = &\sum_{q=0}^{k-1}\sum_{0 \leq i_0<...<i_{r-q-1}\leq m-1}A(m,\xs,2^{m-r+k}) \notag \\
 &+ \sum_{0 \leq i_0<...<i_{r-k-1}\leq m-1}N(m,\xk,2^{m-r+k}).
\end{align}
Next, we analyze $N(m,\xk,2^{m-r+k})$ and $A(m,\xs,2^{m-r+k})$, where $0 \leq q < k$, and then provide asymptotic analysis on $E(A(2^{m-r+k},T))$.

\begin{theo}
\label{th3}
Let $k$ be a positive integer, 
\begin{align}
\label{eql2}
\log N(m,\xk,2^{m-r+k}) \leq  \notag \\
\begin{cases}
0 & i_{r-k-1} \geq r+3, and \\ 
& r \ \text{sufficiently large}; \\
(2^{k+2}-2)r+O(1) & i_{r-k-1} = r+2; \\
(2^{k+2}-1)r & i_{r-k-1} \leq r+1.
\end{cases}
\end{align}
Assume $ m-r \geq 2$, $\frac{r}{m}>\gamma$, where $\gamma > 0$ is a constant. \\Let $0\leq q < k$,
\begin{align}
\label{eql3}
\log A(m,\xs,2^{m-r+k}) \leq  \notag \\
\begin{cases}
0 & i_{r-q-1} \geq r+3, and \\ 
& r \ \text{sufficiently large}; \\
(2^{k+2}-1)r+\log r + O(1) & i_{r-q-1} \leq r+2. 
\end{cases}
\end{align}
Therefore, we have
\begin{align}
\label{eq11r}
(2^{k+2}-1)(r-k) & \leq \log E\left(  A(2^{m-r+k},T) \right) \notag \\ 
&\leq (2^{k+2}-1)r+O(\log r).
\end{align}
\end{theo}
\begin{proof}
The proof is in Appendix \ref{ap4}. The method is similar to that in Theorem \ref{th1}, but due to the sum terms in the recursive formula \cite[Theorem 2]{s2}, the analysis is more complicated. In particular, we derive the upper bound on $A(m,\xs,2^{m-r+k})$ when $i_{r-q-1} \leq r+2$ through induction. 
\end{proof}

In Fig. \ref{fig3}, we display the number of codewords with Hamming weight $2d_{min}$ in RM codes and pre-transformed RM codes on the logarithm domain. The example has code rate $R=0.5$, and the average number of codewords is approximately  $2^{(2^{k+2}-1)r}$ in the order sense. Similarly, the logarithm scaling of the weight-$2d_{min}$ codewords grows linearly with $m$ under pre-transformation, as opposed to quadratically without pre-transformation.
Our approximation is accurate asymptotically, and there is a gap between the true number and approximation when $m$ is small. Note that calculating the accurate number of weight-$2d_{min}$ codewords becomes intractable when $m$ is large.
\begin{figure}[htbp]
\centerline{\includegraphics[width = .32\textwidth]{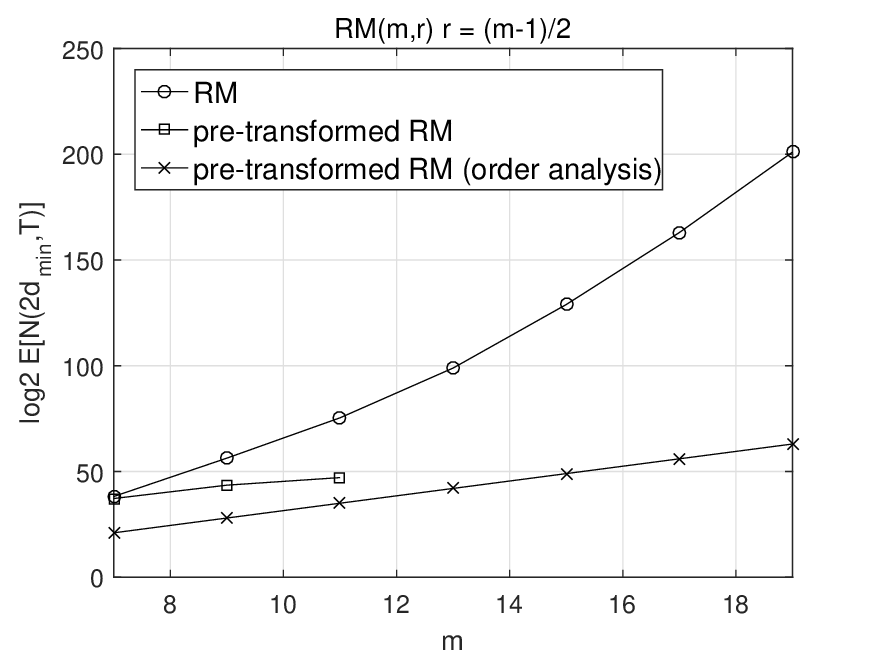}}
\caption{Logarithm scaling of the number of codewords with Hamming weight $2d_{min}$ in RM codes and pre-transformed RM codes.}
\label{fig3}
\end{figure}

\subsection{Minimum-weight codewords of pre-transformed polar codes}
In this section, we extend our analysis from RM codes to polar codes. We  prove that the average number of minimum-weight codewords of pre-transformed polar codes does not increase after pre-transformation.
Unlike RM codes, polar codes do not have a universal sub-channel selection rule. Therefore, their corresponding asymptotic results cannot be obtained as in RM codes. Fortunately, the conclusions in this section are non-asymptotic and apply to arbitrary code lengths.

Let $\CA$ be a decreasing polar code, define
\begin{align}
\label{eqp1}
\overline{r} = \min\{r \ | \ \CA \subseteq \text{RM}(m,r)\},
\end{align}
i.e., the largest degree of monomials in $\mathcal{I}$ is $\overline{r}$ and the minimum distance is $2^{m-\overline{r}}$. According to \cite[Proposition 7]{b21}, the number of minimum-weight codewords of $\CA$ is
\begin{align}
\label{eqp2}
N(2^{m-\overline{r}}) = \sum_{x_{i_{\overline{r}-1}} \cdots x_{i_{0}} \in \mathcal{I}}2^{\sum\limits_{s=0}^{\overline{r}-1}i_s-s+1}.
\end{align}
Similarly, we call $2^{\sum\limits_{s=0}^{\overline{r}-1}i_s-s+1}$ the number of minimum-weight codewords induced by $x_{i_{\overline{r}-1}} \cdots x_{i_{0}}$ in original polar codes. Let $i^*$ be the smallest index in information bits which can be represented by a monomial with degree $\overline{r}$, i.e.
\begin{align}
\label{eqp3}
i^*  = \min\{I_{x_{i_{\overline{r}-1}} \cdots x_{i_{0}}}|x_{i_{\overline{r}-1}} \cdots x_{i_{0}} \in \mathcal{I}\}. 
\end{align}
Next, we prove that pre-transformation does not increase the average number of minimum-weight codewords.
\begin{theo}
\label{th2}
If $\CA$ is a decreasing polar code, $\overline{r}$, $i^*$ are defined in (\ref{eqp1}), (\ref{eqp3}), we have
\begin{align}
\label{eqp4}
E(N(2^{m-\overline{r}},T)) \leq N(2^{m-\overline{r}}). 
\end{align}
Let the monomial representation of the $i^*$-$th$ row vector be \\$x_{i^*_{\overline{r}-1}}\cdots x_{i^*_0}$. If $\overline{r} \leq 1$, (\ref{eqp4}) must hold as an equality. \\
If $\overline{r} > 1$, (\ref{eqp4}) holds as an equality if and only if the following two conditions satisfy:
\begin{align*}
(*)&  \  \ i^*_{\overline{r}-1} \leq \overline{r} + 1. \\
(**)&  \  If \ I_f > i^* \ and \ deg(f) \leq \overline{r}, then \ f \in \mathcal{I}.
\end{align*}
\end{theo}
\begin{proof}
The proof is in Appendix \ref{ap5}. In fact, the number of minimum-weight codewords induced by every $x_{i_{\overline{r}-1}} \cdots x_{i_{0}} \in \mathcal{I}$ decreases after pre-transformation.
\end{proof}
\begin{re}
In fact, the number of minimum-weight codewords induced by every $x_{i_{\overline{r}-1}} \cdots x_{i_{0}} \in \mathcal{I}$ decreases after pre-transformation, but the amount of reduction differs. According to (\ref{eqp2}), the minimum-weight codewords in original polar codes are mainly induced by $x_{i_{\overline{r}-1}} \cdots x_{i_{0}}$ with large $\sum_{s=0}^{\overline{r}-1}(i_s-s)$. Since $i_s-s \leq i_{\overline{r}-1}-(\overline{r}-1)$, $0 \leq s \leq \overline{r}-1$, these monomials also have large $i_{\overline{r}-1}$. As explained in Theorem \ref{th1}, these codewords are reduced due to pre-transformation, which explains why pre-transformation can improve the weight spectrum.  Therefore, more $x_{i_{\overline{r}-1}} \cdots x_{i_{0}}$ with large $i_{\overline{r}-1}$ as information bits results in more significant improvement in the weight spectrum. The results prove that pre-transformation can be beneficial for polar codes too.
\end{re}

\section{conclusion}
In this paper, we provide asymptotic analysis on the number of low-weight codewords of the original and pre-transformed RM codes respectively, and prove that pre-transformation can reduce the low-weight codewords significantly. For decreasing polar codes, we prove that pre-transformation does not increase the average number of minimum-weight codewords. The numerical results validate the theoretical analysis and confirm the benefit of pre-transformation.
\enlargethispage{-1.4cm} 



\newpage
\appendices

\section{}
\label{ap1}
Codewords in polar codes can be regarded as polynomials. Let $f,g$ be two codewords in RM$(m,r)$, define $d(f,g)$ be the Hamming distance between $f$ and $g$.
Define $\operatorname{poly}(d) = \{f \in \text{RM}(m, r) : \omega(f) \leq d\}$, i.e., polynomials in RM$(m, r)$ with weight no larger than $d$. Define the degree of a polynomial $f$ to be the maximal degree of monomials in $f$, denoted  
by $deg(f)$. Let $f$ be a polynomial and $y \in \mathbf F_2^m$. Define the derivative of $f$ in direction $y$ by
\begin{align}
\Delta_y f(x)=f(x+y)+f(x),
\end{align}
we have $deg(\Delta_y f) \leq deg(f)-1$.
For example, let $m=2$, $f(x) = x_0  + x_0x_1$, $y = ( 1, 0)$. Then 
\begin{align}
&\Delta_y f(x)=f(x+y)+f(x) \notag \\ 
=& (x_0+1) + (x_0+1)x_1 + x_0  + x_0x_1 = x_1 + 1. 
\end{align}
Here, $deg(\Delta_y f) = 1 = deg(f)-1$. Define the $k$-iterated derivative of $f$ in direction $Y=\left(y_1, \ldots, y_k\right)\in (\mathbf F_2^{m})^k$ by
\begin{align}
\Delta_Y f(x)=\Delta_{y_1} \Delta_{y_2} \cdots \Delta_{y_k} f(x) .
\end{align} 
Since $deg(\Delta_y f)\leq deg(f)-1$, $deg(\Delta_Y f)\leq deg(f)-k$.

Let $\mathcal{S} \subseteq \operatorname{RM}(m, r)$ be a subset of polynomials, we call a subset of polynomials $\mathcal{B}$ is a $\delta$-net for $\mathcal{S}$ if $\forall f \in \mathcal{S}$, there exists $g \in \mathcal{B}$ such that $d(f, g) \leq \delta$.

\begin{Le}\cite[Corollary 3.1]{y1} \label{lemma1}
Let $t$ be an integer, define
\begin{align*}
\begin{array}{l}
\mathcal{A}_{r-k-1, t}= \{\operatorname{Maj} \left(\Delta_{Y_1} f, \ldots, \Delta_{Y_{t}} f\right):  \\ 
\quad \quad \quad \quad \quad \ \  Y_1, \ldots, Y_{t} \in \left(\mathbf F_2^{m}\right)^{r-k-1}, f \in \operatorname{RM}(m, r)\},
\end{array}
\end{align*}
where $\operatorname{Maj}$ is the majority function defined in \cite{y1}. Then $\mathcal{A}_{r-k-1, t}$ is a $\delta$-net for $\operatorname{poly}(2^{m-r+k})$, where $t=\lceil 17 \log (2^m / \delta)\rceil$.
\end{Le}
\begin{proof}[Proof of Theorem \ref{thy}]
Denote $h = m-r+k$. To prove the lower bound, assume $g(x_0,...,x_{h-1})$ be an arbitrary polynomial with degree $k+1$. Define
\begin{align*}
f\left(x_0, \ldots, x_{m-1}\right)=\left(g\left(x_0, \ldots, x_{h-1}\right)+x_{h}\right)x_{h+1} \ldots x_{m-1}.
\end{align*}
It is clear that $f\in  \operatorname{RM}(m, r)$ and $w(f)=2^h$. The number of polynomials with $h$ variables and degree $k+1$ is
$2^{\binom{h}{k+1}} =2^{\Theta(m^{k+1})}$,
which implies the lower bound.

To prove the upper bound, let $\delta = 2^{m-r-2}$, define $\operatorname{adj}(f)=\{g \in \mathcal{A}_{r-k-1, t}: d(f,g) \leq \delta\}$, where $t=17(r+2)$. By Lemma \ref{lemma1}, $\mathcal{A}_{r-k-1, t}$ is a $\delta$-net for $\operatorname{poly}(2^h)$, thus $\forall f \in \operatorname{poly}(2^h)$, $\operatorname{adj}(f) \neq \emptyset$. Next, we prove for any two differernt $f_1, f_2 \in \operatorname{poly}(2^h)$, $\operatorname{adj}(f_1) \bigcap \operatorname{adj}(f_2)=\emptyset$, otherwise there exist $g \in \mathcal{A}_{r-k-1, t}$, such that $d(f_1, g) \leq \delta$ and  $d(f_2, g) \leq \delta$. By triangle inequality, $d(f_1,f_2)\leq 2^{m-r-1} < d_{min} = 2^{m-r}$, which is a contradiction. Notice that $deg(\Delta_{Y} f) \leq k+1$, $\forall \ Y \in \left(\mathbf F_2^{m}\right)^{r-k-1}$, $f \in \text{RM}(m,r)$, we have
\begin{align}
A(2^h) &\leq \sum\limits_{f \in \tiny{\operatorname{poly}(2^h)}}|\operatorname{adj}(f)| = |\bigcup_{f \in \operatorname{poly}(2^h)}\operatorname{adj}(f)| \notag \\
&\leq |\mathcal{A}_{r-k-1, t}| \leq 2^{t\sum\limits_{s=0}^{k+1}\begin{pmatrix} m \\ s \end{pmatrix}} = 2^{\Theta(m^{k+2})}.
\end{align}
\end{proof}

\section{}
\label{ap2}
\begin{proof}[Proof of Lemma \ref{le1}]
We prove Lemma \ref{le1} via induction on $m$. Firstly, if $m=1$, Lemma \ref{le1} can be proved directly. For the induction step $m-1 \rightarrow m$, we consider two cases according to $i_{r-1}$:

1) $i_{r-1}=m-1$, i.e., $\xir$ is in the top half of $F_N$. The number of information bits in the top half and  whose indices are greater than $I_{\xir}$ is equal to the number of information bits  whose indices are greater than $I_{x_{i_{r-2}}\cdots x_{i_0}}$ in RM$(m-1,r-1)$, which is $\sum\limits_{s=0}^{r-2}\sum\limits_{t=0}^{s+1}\zuhe$ by inductive hypothesis. The number of information bits in the lower half is $\sum\limits_{t=0}^{r}\begin{pmatrix} m-1 \\ t \end{pmatrix}=\sum\limits_{t=0}^{r}\begin{pmatrix} i_{r-1} \\ t \end{pmatrix}$, thus the total number of information bits  whose indices are greater than $I_{\xir}$ is equal to $\sum\limits_{s=0}^{r-1}\sum\limits_{t=0}^{s+1}\zuhe$.
According to \cite[Theorem 1]{s2}, 
\begin{align}
\label{eq5}
&\log P(m,\xir,2^{m-r}) \notag \\
=& \log P(m-1,x_{i_{r-2}}\cdots x_{i_0},2^{m-r}) + 2^{m-r}-2^{m-1}  \notag \\
=& \log P(m-1,x_{i_{r-2}}\cdots x_{i_0},2^{m-r}) + 2^{i_{r-1}-(r-1)}-2^{i_{r-1}}  \notag \\
=&  \sum_{s=0}^{r-1}\left(2^{i_s-s}-2^{i_s}\right).
\end{align}
The last equality is due to inductive hypothesis.

2) $i_{r-1}<m-1$, i.e., $\xir$ is in the lower half of $F_N$. The number of information bits  whose indices are greater than $I_{\xir}$ is equal to the number of information bits  whose indices are greater than $I_{\xir}$ in RM$(m-1,r)$, which is $\sum\limits_{s=0}^{r-1}\sum\limits_{t=0}^{s+1}\zuhe$ by inductive hypothesis.
According to \cite[Theorem 1]{s2},
\begin{align*}
&\log P(m,\xir,2^{m-r}) \notag \\
=& \log P(m-1,x_{i_{r-1}}\cdots x_{i_0},2^{m-1-r})    \notag \\
=&  \sum_{s=0}^{r-1}\left(2^{i_s-s}-2^{i_s}\right).
\end{align*}
The last equality is due to inductive hypothesis.
\end{proof}

\section{}
\label{ap3}

\begin{proof}[Proof of Theorem \ref{th1}] 
Firstly, we analyze $\MN(i_s,s)$ with respect to $s$, in fact, if $s \geq 1$, $i_s-s \geq 3$, $\MN(i_s,s) \leq 0$ for sufficiently large $i_s$.

\emph{case 1}: $s = 0$, $\MN(i_0,0) = 1 + i_0$.

\emph{case 2}: $s = 1$, $\MN(i_1,1) = -2^{i_1-1} + 1 + i_1 + {\begin{pmatrix} i_1 \\ 2 \end{pmatrix}}$, and we have $\MN(i_1,1) \leq 0$ if $i_1 \geq 5$. 

\emph{case 3}: $2 \leq s \leq \lceil \frac{i_s}{2} \rceil-2$,
\begin{align}
\label{eq9}
\MN(i_s,s) &\leq 2^{i_s-2}-2^{i_s}+ \sum_{t=0}^{\lceil \frac{i_s}{2} \rceil - 1}\zuhe \notag \\
&\leq -\frac{3}{4}2^{i_s}+ \frac{1}{2}2^{i_s} = -\frac{1}{4}2^{i_s} \leq 0.
\end{align}

\emph{case 4}: $\lceil \frac{i_s}{2} \rceil-2 \leq s \leq i_s - \log (i_s+16\sqrt{2i_s})$,

\begin{align}
\label{eq10}
 &\MN(i_s,s) = 2^{i_s-s} - \sum\limits_{t=0}^{i_s-s-2}\zuhe \leq 2^{i_s-s}- \begin{pmatrix} i_s \\ i_s-s-2 \end{pmatrix} \notag \\
 \overset{(a)}{\leq}& 2^{i_s-s} - \frac{2^{i_sh(\frac{i_s-s-2}{i_s})}}{\sqrt{2i_s}} \overset{(b)}{\leq} 2^{i_s-s} - \frac{2^{2(i_s-s-2)}}{\sqrt{2i_s}}\notag \\
 = &(1-\frac{2^{i_s-s}}{16\sqrt{2i_s}})2^{i_s-s}\notag \leq (1-\frac{i_s+16\sqrt{2i_s}}{16\sqrt{2i_s}})2^{i_s-s} \notag \\
 \leq & -\frac{i_s^{\frac{3}{2}}}{16\sqrt{2}}-i_s \leq -\frac{i_s^{\frac{3}{2}}}{16\sqrt{2}} \leq 0,
\end{align}
where ($a$) is from  \cite[problem 5.8]{r5}, ($b$) is due to $\frac{i_s-s-2}{i_s} \leq \frac{1}{2}$ when $s \geq \lceil \frac{i_s}{2} \rceil-2$ and $h(x) \geq 2x$, $0 < x \leq \frac{1}{2}$.

\emph{case 5}: $i_s - \log (i_s+16\sqrt{2i_s}) \leq s \leq i_s - 4$,
\begin{align}
\label{eq12}
\MN(i_s,s) &\leq i_s+16\sqrt{2i_s} - \sum\limits_{t=0}^{2}\zuhe \notag \\
& = 16\sqrt{2i_s} - 1 - \frac{i_s(i_s-1)}{2} \overset{(c)}{\leq} 0, 
\end{align}
$(c)$ holds when $i_s \geq 14$.

\emph{case 6}: $s = i_s  - 3$, $\MN(i_s,i_s-3) = 7 - i_s \overset{(d)}{\leq} 0$, $(d)$ holds when $i_s \geq 7$.

\emph{case 7}: $i_s  - 2 \leq s \leq i_s$, $\MN(i_s,s) \leq 3$.

We conclude that $\MN(i_s,s) \leq 0$ if $i_s \geq 14$, $i_s-s \geq 3$, $s \geq 1$ from the discussion above. When $i_s \leq 13$, $s \geq 1$, through compute search, we have $\MN(i_s,s) \leq 3$. Therefore,
\begin{align} 
\label{eq13}
\MN(i_s,s) \leq 3, \ \text{if} \ s \geq 1.
\end{align}
Next, we analyze $\log N(m,\xir,2^{m-r})=\sum\limits_{s=0}^{r-1} \MN(i_s,s)$ with respect tof $i_{r-1}$.

1) $i_{r-1} \leq r+1$, we have $i_s-s \leq i_{r-1}-(r-1) \leq 2$, by \emph{case 7}, 
\begin{align}
\label{eq14}
\log N(m,\xir,2^{m-r}) \leq 3r.
\end{align}

2) $i_{r-1} = r+2$, we have $i_0 \leq i_{r-1}-(r-1) \leq 3$, thus $\MN(i_0,0) = 1+i_0 \leq 4$, by (\ref{eq13}),
 \begin{align}
\label{eq15}
& \log N(m,\xir,2^{m-r})   \notag \\
=& \MN(r+2,r-1) + \sum_{s=1}^{r-2}\MN(i_s,s)+ \MN(i_0,0) \notag \\
\leq & 7-(r+2)+3(r-2)+4  \notag \\
= & 2r+3.
\end{align}

3) $i_{r-1} \geq r+3$, according to \emph{case 3-5}, when $r$ is sufficiently large, we have
\begin{align}
\label{eq16}
\MN(i_{r-1},r-1) \leq -\frac{i_{r-1}^{\frac{3}{2}}}{16\sqrt{2}} \leq -\frac{r^{\frac{3}{2}}}{16\sqrt{2}}. 
\end{align}
Since $i_0 \leq i_1-1$,
\begin{align}
\label{eq17}
&\MN(i_1,1)+\MN(i_0,0) = -2^{i_1-1} + 2 + i_1 + i_0 + {\begin{pmatrix} i_1 \\ 2 \end{pmatrix}}  \notag \\
  \leq &-2^{i_1-1} + 1 + 2i_1 + {\begin{pmatrix} i_1 \\ 2 \end{pmatrix}} \leq 7,
\end{align}
we prove the last inequality through computer search. Thus
\begin{align}
\label{eq18}
& \log N(m,\xir,2^{m-r})  \notag \\
\leq& -\frac{r^{\frac{3}{2}}}{16\sqrt{2}} + 3(r-3)+7  \overset{(e)}{\leq}  -\frac{r^{\frac{3}{2}}}{32} \leq 0,
\end{align}
$(e)$ holds when $r$ is sufficiently large. Thus (\ref{eq7}) is proved from (\ref{eq14}) (\ref{eq15}) (\ref{eq18}).

Now we are ready to prove (\ref{eq11}). On the one hand, $x_{r+1} \cdots x_{2} \in \text{RM}(m,r)$ when $m-r \geq 2$,  therefore
\begin{align}
\label{eq19}
 \log E(N(2^{m-r},T)) \geq \log N(m,x_{r+1} \cdots x_{2},2^{m-r}) =3r.
\end{align}
On the other hand,
\begin{align}
\label{eq20}
&\sum_{i_{r-1} \leq r+1} N(m,\xir,2^{m-r})  \notag \\
& \leq |\{\xir, i_{r-1} \leq r+1\}|2^{3r} = \begin{pmatrix} r+2 \\ 2 \end{pmatrix}2^{3r}.
\end{align}
\begin{align}
\label{eq21}
&\sum_{i_{r-1} = r+2} N(m,\xir,2^{m-r})  \\
& \leq |\{\xir, i_{r-1} = r+2\}|2^{2r+3} = \begin{pmatrix} r+2 \\ 3 \end{pmatrix}2^{2r+3} .\notag 
\end{align}
\begin{align}
\label{eq22}
&\sum_{i_{r-1} \geq r+3} N(m,\xir,2^{m-r})  \notag \\
& \overset{(f)}{\leq} |\{\{\xir, i_{r-1} \geq r+3\}|2^{-\frac{r^{\frac{3}{2}}}{32}} \leq 2^{m-\frac{r^{\frac{3}{2}}}{32}} \notag \\
& \leq 2^{\frac{r}{\gamma}-\frac{r^{\frac{3}{2}}}{32}} \overset{(g)}{<} 1,
\end{align}
where $(f)$ holds if (\ref{eq18}) holds, $(g)$ holds if $r \geq \left(\frac{32}{\gamma}\right)^2$. Divide the sum terms in (\ref{eq2}) into three parts according to $i_{r-1}$, by (\ref{eq20})-(\ref{eq22}),
\begin{align}
& E(N(2^{m-r},T)) \notag \\
\leq &\begin{pmatrix} r+2 \\ 2 \end{pmatrix}2^{3r} + \begin{pmatrix} r+2 \\ 3 \end{pmatrix}2^{2r+3} + 1 \notag \\
 = &\begin{pmatrix} r+2 \\ 2 \end{pmatrix}2^{3r}(1+o(1)),  \\
& \log E(N(2^{m-r},T)) \leq 3r + \log \begin{pmatrix} r+2 \\ 2 \end{pmatrix} + o(1).  \label{eq24}
\end{align}
Combine (\ref{eq19}) and (\ref{eq24}), we complete the proof of (\ref{eq11}).
\end{proof}

\section{}
\label{ap4}
\emph{Proof outline:}

Firstly, we prove (\ref{eql2}) by \emph{Step 1-2}, the proof of (\ref{eql2}) is similar to that of (\ref{eq7}), and is omitted due to space limitation.

\emph{Step 1}, let
\begin{align}
\label{eq31}
\mathcal{N}_k(i_s,s)  &= 2^{i_s-s}-2^{i_s}+\sum\limits_{t=0}^{s+k+1}\zuhe \notag \\
& = 2^{i_s-s}-\sum\limits_{t=0}^{i_s-s-k-2}\zuhe,
\end{align}
we prove $\mathcal{N}_k(i_s,s) \leq 0$ if $s \geq 1$, $i_s-s \geq k+3$.

\emph{Step 2}, similar to (\ref{eq4}), we have 
\begin{align}
\label{eq32}
\log N(m,\xk,2^{m-r+k})=\sum\limits_{s=0}^{r-k-1} \MN_k(i_s,s),
\end{align} 
and we analyze $\log N(m,\xk,2^{m-r+k})$ with respect to $i_{r-k-1}$. 

Next, we prove (\ref{eql3}) by \emph{Step 3-5}.

\emph{Step 3}, let
\begin{align}
\label{eqd1}
&\MN_{k,q}(i_s,s) \overset{def}{=}  \notag \\
&\max_{2^{i_s-s}\leq d' \leq 2^{i_s-s+k-q}}(2^{i_s}-d')(h(\frac{2^{i_s-s+k-q}-d'}{2(2^{i_s}-d')})-1) \notag \\
&+ \sum_{t=0}^{q+s+1} \zuhe +i_s-s+ k-q,   
\end{align}
where $s \geq k-q+1$, by (\ref{eqd5}) and (\ref{eqd6}), we prove
\begin{align}
\label{eqd2}
&\log A(m,\xs,2^{m-r+k}) \notag \\ 
\leq & m +  \sum_{s=0}^{k-q}\left(\sum_{t=0}^{q+s+1} \zuhe + i_s\right)
+ \sum_{s=k-q+1}^{r-q-1}\MN_{k,q}(i_s,s)    
    \end{align}
via induction on $m$.

\emph{Step 4}, we prove $\mathcal{N}_{k,s}(i_s,s) \leq 0$ when $i_s-s$ is large. Therefore, if $i_{r-q-1} \geq r+3$, we have
$\log A(m,\xs,2^{m-r+k}) \leq 0$.

The proof of \emph{Step 3-4} is omitted due to space limitation.

\emph{Step 5}, when $i_{r-q-1} \leq r+2$, we prove 
\begin{align*}
&\log A(m,\xs,2^{m-r+k}) \\
\leq &(2^{k+2}-1)r+ \log r + O(1)
\end{align*}
via induction.

Finally, we prove (\ref{eq11r}) by \emph{Step 6}.

\emph{Step 6}, based on \emph{Step 4-5}, we have
\begin{align}
\label{eqd3}
& E(A(2^{m-r+k},T)) \notag \\
 \approx &\sum_{q=0}^{k-1} \sum_{i_{r-q-1} \leq r+2} A(m,\xs,2^{m-r+k})\notag \\
& + \sum_{i_{r-k-1} \leq r+2} N(m,\xk,2^{m-r+k}).  
\end{align}
Combine (\ref{eql2}) and (\ref{eql3}), we complete the proof of (\ref{eq11r}).

\begin{proof}[Proof of Theorem \ref{th3}]
$\forall \ 0 \leq q < k$, \\ $N(m,\xs,d)$ is the number of weight-$d$ codewords induced by $\xs$ in RM$(m,r)$, we have 
\begin{align}
\label{eqd4}
&\log N(m,\xs,d) \notag \\ 
= &\log P(m,\xs,d) + \sum_{s=0}^{r-q-1}\sum_{t=0}^{s+q+1}\zuhe.
\end{align}
According to \cite[Theorem 2]{s2}, if $i_{r-q-1}=m-1$, 
\begin{align}
\label{eqd5}
&N(m,\xs,d)  \notag \\
=&\sum\limits_{\mbox{\tiny$\begin{array}{c}
d'= 2^{m-r+q}\\
d-d' is \ even \end{array}$}}^{d}\Bigg(N(m-1,x_{i_{r-q-2}}\cdots x_{i_0},d')* \notag \\
&2^{\Bigg(d'-\sum\limits_{t=0}^{m-r-2} \begin{pmatrix} m-1 \\ t \end{pmatrix}\Bigg)}*\begin{pmatrix} 2^{m-1}-d' \\ \frac{d-d'}{2} \end{pmatrix}\Bigg).
\end{align}
If $i_{r-q-1} < m-1$, 
\begin{align}
\label{eqd6}
&N(m,\xs,d) = N(m-1,\xs,\frac{d}{2}).
\end{align}
We are going to prove $\log A(m,\xs,2^{m-r+k}) \leq (2^{k+2}-1)r+\log r + O(1)$ when $i_{r-q-1} \leq r+2$.

Let $i_{r-q-1} = r-q-1+\ell$, where $0 \leq \ell \leq q+3$, apply (\ref{eqd6}) $m-r+q-\ell$ times repeatedly, we have 
\begin{align}
\label{eq3r21}
&N(m,\xs,d) \notag \\
= &N(r-q+\ell,\xs,\frac{d}{2^{m-r+q-\ell}}),
\end{align}
where $2^{m-r+q} \leq d \leq 2^{m-r+k}$. (\ref{eq3r21}) must be zero unless $\frac{d}{2^{m-r+q-\ell}}=2^\ell+2v$, where $v$ is a non-negtive integer, therefore \\
$N(m,\xs,d) =$
\begin{align}
\label{eq3r22}
\begin{cases}
N(r-q+\ell,\xs,2^{\ell}+2v)&   \\ 
\qquad \qquad  \quad \ \     d = (2^\ell+2v)2^{m-r+q-\ell}, v \geq 0; \\
0  \qquad \qquad \qquad \qquad \qquad \qquad  \quad \ otherwise. 
\end{cases}
\end{align}
Next, let $C_0 = 1, C_v = 2^{v-1}, v \geq 1$, we are going to prove that if $i_{r-q-1} = r-q-1+\ell$, $0 \leq \ell \leq q+1$,
\begin{align}
\label{eq3r23}
&N(r-q+\ell,\xs,2^{\ell}+2v) \notag \\ 
\leq &C_v(r-q+\ell)2^{(2^\ell+2v)(r-q+\ell)}
\end{align}
via induction on $\ell$, $v$ and $r-q$, the degree of the monomial $\xs$. 

When $\ell=0$, we prove (\ref{eq3r23}) holds via induction on $v$ and $r-q$.
When $\ell \geq 1$, in addition to induction on $v$ and $r-q$, we also use the inductive hypothesis that (\ref{eq3r23}) holds from $0$ to $\ell-1$, $\forall \ v \geq 0, r-q \geq 0$.

If $v = 0$, by (\ref{eq32}),
\begin{align}
\label{eq3r24}
&\log N(r-q+\ell,\xs,2^{\ell}) = \sum\limits_{s=0}^{r-q-1} \MN_q(i_s,s) \notag \\   
 = &\sum\limits_{s=0}^{r-q-1}\left(2^{i_s-s}-\sum\limits_{t=0}^{i_s-s-q-2}\zuhe\right) \leq 2^\ell(r-q),
\end{align}
where the last inequality is due to $i_s-s \leq i_{r-q-1}-(r-q-1) = \ell$. Therefore,
\begin{align}
\label{eq3r24p} 
N(r-q+\ell,\xs,2^{\ell}) \leq C_0(r-q+\ell)2^{2^\ell(r-q+\ell)}.
\end{align}

For the induction step $v-1 \rightarrow v$, denote $r-q=n$ for convience, we complete the induction step via induction on $n$. 

When $n = 0$, $v \geq 1$, no codeword has Hamming weight $2^\ell+2v$ which is larger than the code length $2^\ell$, thus
\begin{align}
\label{eq3r25}
N(\ell,\mathbf{1},2^{\ell}+2v) = 0 \leq C_v\ell2^{(2^\ell+2v)\ell},
\end{align}
where $\mathbf{1}$ represents the monomial with degree $0$. 

For the induction step $n-1 \rightarrow n$, by (\ref{eqd5}),
\begin{align}
&N(n+\ell,x_{i_{n-1}}\cdots x_{i_0},2^\ell+2v) \notag \\
=& \sum_{\mu=0}^{v} \Bigg(N(n-1+\ell,x_{i_{n-2}}\cdots x_{i_0},2^\ell+2\mu)* \notag \\
&2^{\Bigg( 2^\ell+2\mu-\sum\limits_{t=0}^{\ell-q-2}\begin{pmatrix} n-1+\ell \\ t\end{pmatrix}\Bigg)}*\begin{pmatrix} 2^{n-1+\ell}-(2^\ell+2\mu) \\ v-\mu\end{pmatrix}\Bigg) \label{eq3r26e} \\
 \overset{(h)}{\leq} &\sum_{\mu=0}^{v}\Bigg( N(n-1+\ell,x_{i_{n-2}}\cdots x_{i_0},2^\ell+2\mu)* \notag \\
&2^{2^\ell+2\mu+(v-\mu)(n+\ell)}\Bigg),  \label{eq3r26d}
\end{align}
where $(h)$ is from $\begin{pmatrix} 2^{n-1+\ell}-(2^\ell+2\mu) \\ v-\mu\end{pmatrix} \leq 2^{(v-\mu)(n+\ell)}$, and $\sum\limits_{t=0}^{\ell-q-2}\begin{pmatrix} n-1+\ell \\ t\end{pmatrix}=0$ since $\ell \leq q+1$.  

If $i_{n-2}=n-2+\ell$, by inductive hypothesis on $n-1$,
\begin{align}
\label{eqdd1}
&N(n-1+\ell,x_{i_{n-2}}\cdots x_{i_0},2^\ell+2\mu) \notag \\
\leq &C_\mu(n-1+\ell)2^{(2^\ell+2\mu)(n-1+\ell)}.
\end{align}

If $i_{n-2}=n-2+\ol$, $\ol < \ell$, apply (\ref{eqd6}) $\ell-\ol$ times repeatedly,
\begin{align}
\label{eqdd2}
&N(n-1+\ell,x_{i_{n-2}}\cdots x_{i_0},2^\ell+2\mu) \notag \\
=&N(n-1+\ol,x_{i_{n-2}}\cdots x_{i_0},2^{\ol}+2\frac{\mu}{2^{\ell-\ol}}).
\end{align}
If $\frac{\mu}{2^{\ell-\ol}}$ is an integer, by inductive  hypothesis on $\ol$,
\begin{align}
\label{eqdd3}
&N(n-1+\ol,x_{i_{n-2}}\cdots x_{i_0},2^{\ol}+2\frac{\mu}{2^{\ell-\ol}}) \notag \\
\leq &C_{\frac{\mu}{2^{\ell-\ol}}}(n-1+\ol)2^{(2^{\ol}+2\frac{\mu}{2^{\ell-\ol}})(n-1+\ol)} \notag \\
\leq &C_\mu(n-1+\ell)2^{(2^\ell+2\mu)(n-1+\ell)}, 
\end{align}
otherwise
\begin{align}
\label{eqdd4}
N(n-1+\ol,x_{i_{n-2}}\cdots x_{i_0},2^{\ol}+2\frac{\mu}{2^{\ell-\ol}})=0.
\end{align}
Combine (\ref{eqdd1})-(\ref{eqdd4}), we have
\begin{align}
\label{eqdd5}
&N(n-1+\ell,x_{i_{n-2}}\cdots x_{i_0},2^\ell+2\mu) \notag \\
\leq &C_\mu(n-1+\ell)2^{(2^\ell+2\mu)(n-1+\ell)}.
\end{align}
Continue the proof in (\ref{eq3r26d}), we have 
\begin{align}
\label{eq3r26}
&N(n+\ell,x_{i_{n-1}}\cdots x_{i_0},2^\ell+2v) \notag \\
\leq &\sum_{\mu=0}^{v-1}C_\mu(n-1+\ell)2^{(2^\ell+v+\mu)(n+\ell)} \notag \\
&+ C_{v}(n-1+\ell)2^{(2^\ell+2v)(n+\ell)} \notag \\
=& \left( \sum_{\mu=0}^{v-1}C_\mu \frac{n-1+\ell}{2^{(v-\mu)(n+\ell)}}+ C_v(n-1+\ell)\right) 2^{(2^\ell+2v)(n+\ell)} \notag \\
\leq &\left(\sum_{\mu=0}^{v-1}C_\mu\frac{ n+\ell}{2^{n+\ell}}+ C_v(n-1+\ell)\right) 2^{(2^\ell+2v)(n+\ell)} \notag \\
\leq &\left(\sum_{\mu=0}^{v-1}C_\mu+ C_v(n-1+\ell)\right) 2^{(2^\ell+2v)(n+\ell)} \notag \\
= &\left(1+\sum_{\mu=1}^{v-1}2^{\mu-1}+ 2^{v-1}(n-1+\ell)\right) 2^{(2^\ell+2v)(n+\ell)} \notag \\
= &C_v(n+\ell)2^{(2^\ell+2v)(n+\ell)},
\end{align}
the induction step $n-1 \rightarrow n$ holds, thus we complete the proof of (\ref{eq3r23}). 

Therefore, when $i_{r-q-1}= r-q-1+\ell$, $\ell \leq q+1$,
\begin{align}
&A(m,\xs,2^{m-r+k}) \notag \\
= &\sum_{d=2^{m-r+q}}^{2^{m-r+k}}N(m,\xs,d) \notag \\
\overset{(i)}{=}& \sum_{v=0}^{2^{k-q+\ell-1}-\lceil 2^{\ell-1} \rceil} N(r-q+\ell,\xs,2^{\ell}+2v) \notag \\
\leq &\sum_{v=0}^{2^{k-q+\ell-1}-\lceil 2^{\ell-1} \rceil}C_v(r-q+\ell)2^{(2^\ell+2v)(r-q+\ell)} \label{eq3r27d} \\ 
\overset{(j)}{\leq}& (2^{k-q+\ell-1}-\lceil 2^{\ell-1} \rceil+1)*  \notag \\
&C_{2^{k-q+\ell-1}}(r-q+\ell)2^{2^{k-q+\ell}(r-q+\ell)} \label{eq3r27d2} \\
\overset{(k)}{\leq}& 2^{k+2}C_{2^{k+2}}(r+3)2^{(2^{k+2}-1)(r+3)}, \label{eq3r27}
\end{align}
where $(i)$ is due to (\ref{eq3r22}). In $(j)$, since the sum terms in (\ref{eq3r27d}) are increasing with respect to $v$, $\forall \ v \leq 2^{k-q+\ell-1}-\lceil 2^{\ell-1} \rceil$,
\begin{align}
\label{eq3r27d1}
&C_v(r-q+\ell)2^{(2^\ell+2v)(r-q+\ell)} \notag \\
\leq &C_{2^{k-q+\ell-1}-\lceil 2^{\ell-1} \rceil}(r-q+\ell)2^{2^{k-q+\ell+2^\ell-2\lceil 2^{\ell-1} \rceil}(r-q+\ell)} \notag \\
\leq &C_{2^{k-q+\ell-1}}(r-q+\ell)2^{2^{k-q+\ell}(r-q+\ell)}. 
\end{align}
$(k)$ is due to $\ell \leq q+1$, and we take (\ref{eq3r27}) as an upper bound on (\ref{eq3r27d2}) independent of $\ell$ for the convenience of the following analysis.

If $i_{r-q-1} = r-q-1+\ell$, $q+2 \leq \ell \leq q+3$, similar results can be proved via induction, due to space limitation, we only provide inductive hypothesis when $q+2 \leq \ell \leq q+3$ without proof. 

If $\ell = q+2$, $i_{r-q-1} = r+1$,
\begin{align}
\label{eq3r31}
&N(r+2,\xs,2^{q+2}+2v) \notag \\ 
\leq &C_v(r+2)2^{(2^{q+2}+2v-1)(r+2)}.
\end{align}
The only difference between the proof of (\ref{eq3r31}) and (\ref{eq3r23}) is that in (\ref{eq3r26e}),  $\sum\limits_{t=0}^{\ell-q-2}\begin{pmatrix} n-1+\ell \\ t\end{pmatrix}=1$ when $\ell = q+2$.

Similar to (\ref{eq3r27}), we have
\begin{align}
\label{eq3r32}
&A(m,\xs,2^{m-r+k}) \notag \\
\leq &2^{k+1}C_{2^{k+1}-2^{q+1}}(r+2)2^{(2^{k+2}-1)(r+2)} \notag \\ 
\leq &2^{k+2}C_{2^{k+2}}(r+3)2^{(2^{k+2}-1)(r+3)}.
\end{align}

If $\ell = q+3$, $i_{r-q-1} = r+2$,
\begin{align}
\label{eq3r34}
&N(r+3,\xs,2^{q+3}+2v) \notag \\ 
\leq &C_v(r+3)2^{(2^{q+2}+v-1)(r+3)}.
\end{align}
The only difference between the proof of (\ref{eq3r34}) and (\ref{eq3r23}) is that in (\ref{eq3r26e}), $\sum\limits_{t=0}^{\ell-q-2}\begin{pmatrix} n-1+\ell \\ t\end{pmatrix}=n+q+3$ when $\ell = q+3$.

Similar to (\ref{eq3r27}), we have 
\begin{align}
\label{eq3r35}
&A(m,\xs,2^{m-r+k}) \notag \\
\leq &2^{k+2}C_{2^{k+2}-2^{q+2}}(r+3)2^{(2^{k+2}-1)(r+3)} \notag \\ 
\leq &2^{k+2}C_{2^{k+2}}(r+3)2^{(2^{k+2}-1)(r+3)}.
\end{align}
Combine (\ref{eq3r27}) (\ref{eq3r32}) (\ref{eq3r35}), if $i_{r-q-1} \leq r+2$,
\begin{align}
\label{eq3r36}
&\log A(m,\xs,2^{m-r+k})  \notag \\
\leq &(2^{k+2}-1)r+ \log r + O(1).
\end{align}
Now we are ready to prove (\ref{eq11r}). On the one hand, $x_{r+1} \cdots x_{k+2} \in \text{RM}(m,r)$ when $m-r \geq 2$, by (\ref{eq32}),
\begin{align}
\label{eq319}
 \log E(A(2^{m-r+k},T)) &\geq \log N(m,x_{r+1} \cdots x_{k+2},2^{m-r+k}) \notag \\
& =(2^{k+2}-1)(r-k).
\end{align}

On the other hand, by (\ref{eql2}) and (\ref{eql3}), similar to (\ref{eq22}), \\ $\sum_{q=0}^{k-1} \sum_{i_{r-q-1} \geq r+3} A(m,\xs,2^{m-r+k})$ and \\ $\sum_{i_{r-k-1} \geq r+3} N(m,\xk,2^{m-r+k})$ are negligible,
\begin{align}
\label{eq323}
& E(A(2^{m-r+k},T)) \notag \\
 \leq &\sum_{q=0}^{k-1} \sum_{i_{r-q-1} \leq r+2} A(m,\xs,2^{m-r+k})\notag \\
& + \sum_{i_{r-k-1} \leq r+2} N(m,\xk,2^{m-r+k}) + O(1) \notag \\
 \leq &\sum_{q=0}^{k-1}\begin{pmatrix} r+3 \\ q+3 \end{pmatrix}2^{k+2}C_{2^{k+2}}(r+3)2^{(2^{k+2}-1)(r+3)} \notag \\
&+\begin{pmatrix} r+3 \\ k+3 \end{pmatrix}2^{(2^{k+2}-1)(r-k)}+ O(1).  
\end{align}
Combine (\ref{eq319}) and (\ref{eq323}), we have
\begin{align}
\label{eq324}
(2^{k+2}-1)(r-k) & \leq \log E\left(  A(2^{m-r+k},T) \right) \notag \\ 
&\leq (2^{k+2}-1)r+O(\log r).
\end{align}
\end{proof}

\section{}
\label{ap5}
\begin{proof}[Proof of Theorem \ref{th2}]
Firstly, we prove 
\begin{align}
\label{eqp6}
\MN(i_s,s)=2^{i_s-s}-2^{i_s}+ \sum\limits_{t=0}^{s+1} \begin{pmatrix} i_s \\ t \end{pmatrix} \leq i_s-s+1.
\end{align}

1) If $s = 0$ or $i_s-s \leq 2$, (\ref{eqp6}) holds as an equality. 

2) If $s \geq 1$ and $i_s-s \geq 3$, by (\ref{eq13}),
\begin{align}
\label{eqp7}
\MN(i_s,s) \leq 3 < i_s-s+1.
\end{align}
Therefore, if $s \geq 1$ and $i_s-s \geq 3$,
\begin{align}
\label{eqp8}
2^{i_s-s}-2^{i_s}+ \sum\limits_{t=0}^{s+1} \begin{pmatrix} i_s \\ t \end{pmatrix} < i_s-s+1.
\end{align}
Now we are ready to prove (\ref{eqp4}), 
\begin{align}
\label{eqp5}
& E(N(2^{m-\overline{r}},T))  \notag \\ 
 = &\sum_{\xx \in \mathcal{I}} 2^{\sum\limits_{s=0}^{\overline{r}-1} \left(2^{i_s-s}-2^{i_s}\right)+ |\{I_f > I_{\xx}|f \in \mathcal{I} \}|}\notag \\
 \leq &\sum_{\xx \in \mathcal{I}} 2^{\sum\limits_{s=0}^{\overline{r}-1} \left(2^{i_s-s}-2^{i_s}\right)+ |\{I_f > I_{\xx}|deg(f) \leq \overline{r} \}|}\notag \\
 = &\sum_{\xx \in \mathcal{I}} 2^{\sum\limits_{s=0}^{\overline{r}-1} \left(2^{i_s-s}-2^{i_s}+ \sum\limits_{t=0}^{s+1} \begin{pmatrix} i_s \\ t \end{pmatrix}\right)} \notag \\
 \leq &\sum_{\xx \in \mathcal{I}} 2^{\sum\limits_{s=0}^{\overline{r}-1} (i_s-s+1)} \notag \\
 = &N(2^{m-\overline{r}}),
\end{align}  
where the first inequality holds since $\overline{r}$ is the largest degree of monomials in $\mathcal{I}$.

When $\overline{r} \leq 1$, the above two inequalities must be equalities, thus (\ref{eqp4}) holds as an equality. 

When $\overline{r} > 1$, 
$$\sum\limits_{s=0}^{\overline{r}-1} (i_s-s+1) =  \sum\limits_{s=0}^{\overline{r}-1} \left(2^{i_s-s}-2^{i_s}+ \sum\limits_{t=0}^{s+1} \begin{pmatrix} i_s \\ t \end{pmatrix}\right)$$
if and only if $i_{\overline{r}-1} \leq \overline{r}+1$, we conclude that (\ref{eqp4}) holds as an equality if and only if $\forall \ \xx \in \mathcal{I}$
\begin{align*}
(*')&  \  \ i_{\overline{r}-1} \leq \overline{r} + 1. \\
(**')&  \  \text{If} \ I_f > I_{\xx} \ \text{and} \ deg(f) \leq \overline{r}, \text{then} \ f \in \mathcal{I}.
\end{align*}
Apparently, $(*')$ is equivalent to $(*)$, and $(**')$ is equivalent to $(**)$, thus the proof is completed.
\end{proof}
\end{document}